\documentclass[a4paper,11pt, reqno]{amsart}
\usepackage[latin1]{inputenc}
\usepackage{mathrsfs}
\usepackage{amsmath}
\usepackage{amssymb}
\usepackage{amsthm}
\usepackage{amsfonts}
\usepackage{amstext}
\usepackage{amsopn}
\usepackage{amsxtra}
\usepackage{mathrsfs}
\newtheorem{theorem}{Theorem}
\newtheorem{proposition}[theorem]{Proposition}

\newtheorem{lemma}[theorem]{Lemma}
\newtheorem{corollary}[theorem]{Corollary}

\newtheorem{definition}{Definition}
\newtheorem{remark}{Remark}

\newcommand{\ii}{\infty}
\newcommand\R{{\ensuremath {\mathbb R} }}
\newcommand\C{{\ensuremath {\mathbb C} }}
\newcommand\N{{\ensuremath {\mathbb N} }}

\newcommand{\gto}{\underset{\rm g}{\rightharpoonup}}
\renewcommand\phi{\varphi}

\newcommand{\gH}{\mathfrak{H}}
\newcommand{\gS}{\mathfrak{S}}

\newcommand{\wto}{\rightharpoonup}

\newcommand{\cF}{\mathcal{F}}

\newcommand{\cFN}{\cF^{\leq N}}
\newcommand{\eps}{\epsilon}

\renewcommand{\epsilon}{\varepsilon}

\newcommand{\norm}[1]{ \left| \! \left| #1 \right| \! \right| }
\DeclareMathOperator{\tr}{{\rm Tr}}
\DeclareMathOperator{\Tr}{{\rm Tr}}
\DeclareMathOperator{\glim}{g-lim}

\renewcommand{\geq}{\geqslant}
\renewcommand{\leq}{\leqslant}
\renewcommand{\hat}{\widehat}
\renewcommand{\tilde}{\widetilde}

\title{A many-body RAGE theorem}

\author{Jonas Lampart}
\email{lampart@ceremade.dauphine.fr}
\address{CNRS and CEREMADE (UMR CNRS 7534), University of Paris-Dauphine, Place de Lattre de Tassigny, 75775 Paris Cedex 16, France}

\author{Mathieu Lewin}
\email{mathieu.lewin@math.cnrs.fr}
\address{CNRS and CEREMADE (UMR CNRS 7534), University of Paris-Dauphine, Place de Lattre de Tassigny, 75775 Paris Cedex 16, France}

\date{July 20, 2015. Final version to appear in \emph{Comm. Math. Phys.}}

\begin{document}

\begin{abstract}
We prove a generalized version of the RAGE theorem for $N$-body quantum systems. The result states that only bound states of systems with $0\leq n\leq N$ particles persist in the long time average. The limit is formulated by means of an appropriate weak topology for many-body systems, which was introduced by the second author in a previous work, and is based on reduced density matrices. This topology is connected to the weak-$\ast$ topology of states on the algebras of canonical commutation or anti-commutation relations, and we give a formulation of our main result in this setting.
\end{abstract}

\maketitle


\section{Introduction and main result}

The RAGE theorem, due to Ruelle~\cite{Ruelle-69}, Amrein-Georgescu~\cite{AmrGeo-73} and Enss~\cite{Enss-78} is a famous result relating the long time behavior of solutions to the Schr{\"o}\-dinger equation and the spectral properties of the corresponding self-adjoint Hamiltonian. In particular, it states that for any fixed $x$ in the ambient Hilbert space $\gH$,
\begin{equation}
 \lim_{T\to\ii}\frac{1}{T}\int_0^T \norm{KP_c\,e^{-itH}x}^2\,dt=0\,,
 \label{eq:usual_RAGE}
\end{equation}
where $P_c$ is the spectral projector to the continuous spectral subspace of $H$, $K$ is any compact operator and $x(t)=e^{-itH}x$ is the unique (weak) solution to Schr\"odinger's equation
$$\begin{cases}
i\,\dot x(t)=Hx(t),\\
x(0)=x.
  \end{cases}
$$
An equivalent way of formulating the same result is as follows: for every positive self-adjoint operator $\gamma$ in the trace class $\gS^1(\gH)$, consider the ergodic mean 
$$M(T):=\frac{1}{T}\int_0^T e^{-itH}\gamma e^{itH}\,dt$$
which is uniformly bounded in $\gS^1(\gH)$.
Then, the projection $P_cM(T)P_c$ converges weakly-$\ast$ to 0 as $T\to\ii$.
 Here, the operator $\gamma(t)=e^{-itH}\gamma e^{itH}$ is the unique (weak) solution to von Neumann's formulation of Schr\"odinger's equation
\begin{equation}\label{eq:Heisenberg}
\begin{cases}
i\,\dot \gamma(t)=[H,\gamma(t)],\\
\gamma(0)=\gamma.
\end{cases}
\end{equation}
The previous formulation~\eqref{eq:usual_RAGE} corresponds to $\gamma=|x\rangle\langle x|$.
 Put differently, any weakly-$\ast$ convergent subsequence of $M(T)$ has a limit $M_\ii$ which is supported on the space spanned by the eigenvectors of $H$. It can also be proved that $M_\ii$ commutes with $H$, that is, 
$$M_\ii = \sum_j \alpha_j |\phi_j\rangle\langle\phi_j|$$
where $\phi_j$ is an orthonormal system of eigenvectors of $H$.

The RAGE theorem is a very important result in quantum mechanics. For an infinite-dimensional Hamiltonian system such as Schr\"odinger's equation, strong convergence to stationary states (that is, eigenvectors of $H$) cannot hold in general, due to the conservation laws. The RAGE theorem states that, on the contrary, \emph{weak} convergence towards stationary states holds. Hence, in this sense, only bound states persist in the long time average. By virtue of its generality, self-adjointness of $H$ being the only hypothesis, the RAGE theorem is a fundamental tool in the spectral theory of self-adjoint operators. For instance, it may serve as a first step towards a more precise scattering theory of quantum systems~\cite{ReeSim3} and it is also often used in the study of Anderson localization~\cite{Graf-94,Hundertmark-00}.

However, the information it provides for an interacting many-body system is often not very precise. Consider for instance three electrons in the field of a proton, described by the three-body Hamiltonian
\begin{equation}
 H_3=\sum_{j=1}^3\left(-\Delta_{x_j}-\frac{1}{|x_j|}\right)+\sum_{1\leq j<\ell\leq 3}\frac{1}{|x_j-x_\ell|}.
\end{equation}
It is known that $H_3$ has no bound state \cite{Hill-80,Lieb-84}. The RAGE theorem therefore tells us that the ergodic mean $M(T)$ tends weakly-$\ast$ to zero for every initial condition $\gamma$ (a trace-class operator on $L^2(\R^9)$). On the other hand, the corresponding Hamiltonian $H_2$ for two electrons has finitely many bound states~\cite{Hill-77,Hill-77b,GroPit-83} and, of course, the hydrogen atom $H_1$ for one electron has infinitely many.  The physical picture is that some of the three particles escape, whereas the rest remain in a bound state of $H_2$ or $H_1$, a phenomenon that is not captured at all by the RAGE theorem. The precise description of this process through scattering theory has been the object of several works~\cite{Derezinski-93,SigSof-94,HunSig-00b}. It was proved that, asymptotically, the time-evolved wavefunction can be approximated by a sum of products of the form $e^{-it\lambda}\psi_t\otimes\phi$ where $\phi$ is a $\lambda$-eigenfunction of $H_1$ or $H_2$ and $\psi_t\wto0$. A 
tensor product of this form tends weakly to zero in $L^2(\R^9)$, which is why the weak-$\ast$ limit of $M(T)$ is always zero. 

In this paper, we would like to prove a new general version of the RAGE theorem that retains some information on the particles that do not escape, without addressing all the details of the scattering process. This is done by using another weak topology, for which the tensor product $\psi_t\otimes\phi$ converges to $\phi$, and which was introduced by the second author in~\cite{Lewin-11}. In this topology, the ergodic mean will converge to states with possibly less particles, which are supported in the point spectrum of the Hamiltonians $H_n$ for $n\leq N$. This captures the principal physical ideas, even in situations where scattering is not known or not believed to hold, for instance for potentials with an arbitrarily slow decay at infinity.

\subsection*{Geometric convergence}
In order to state our main result, let us quickly describe the notion of convergence used in~\cite{Lewin-11}, where all the details may be found. Let $\gH$ be any separable Hilbert space and denote by $\gH^N$ the symmetric (or antisymmetric) $N$-fold tensor product, respectively denoted by $\otimes_{s/a}$, of $\gH$. Let $\Gamma_k$ be a sequence of $N$-particle states, that is, $\Gamma_k\geq0$ and $\tr(\Gamma_k)=1$. For instance, for a pure state $\Gamma_k=|\Psi_k\rangle\langle\Psi_k|$ for a normalized $\Psi_k\in\gH^N$. The \emph{$n$-particle density matrix} of $\Gamma_k$ is obtained by taking the partial trace with respect to $N-n$ variables and multiplying by an appropriate normalization constant:
$$\Gamma_k^{(n)}={N\choose n}\tr_{n+1,...,N}(\Gamma_k).$$
The sequence $\Gamma_k$ is said to \emph{converge geometrically} to a state $\Gamma$ if the reduced density matrices of $\Gamma_k$ all converge weakly-$\ast$ to those of $\Gamma$. Except if convergence holds in trace-norm, the state $\Gamma$ can never be an $N$-particle state. It is necessary to work with states on the truncated Fock space
$$\cF^{\leq N}(\gH):=\C\oplus \gH\oplus\cdots\oplus \gH^N$$
and this corresponds to the picture that some particles can be lost. 
For simplicity, all the states we consider in this paper are assumed to commute with the particle number, the theory for the general case is essentially the same. Such states can be written in block form as $\Gamma=G_0\oplus\cdots \oplus G_N$, and if we start with a sequence of $N$-body states and investigate its geometric limits, these are the only states that can be obtained.

We rephrase the previous discussion in the following:

\begin{definition}[Geometric convergence]
A sequence of states $\{\Gamma_k\}_{k=1}^\ii$ on $\gH^N$ \emph{converges geometrically} to a state $\Gamma=G_0\oplus\cdots \oplus G_N$ on $\cFN(\gH)$, if $$\Gamma_k^{(n)}\wto_\ast\Gamma^{(n)}=G_n+\sum_{m=n+1}^N{m\choose n}\tr_{n+1,...,m}(G_m)$$ 
weakly-$\ast$ in the trace-class $\gS^1(\gH^n)$ for all $n=0,...,N$. That is, $\tr(K\Gamma_k^{(n)})\to\tr(K\Gamma^{(n)})$ for every compact operator $K$ on $\gH^n$. This notion of convergence is denoted as $\Gamma_k\wto_{\rm g}\Gamma$ and extended by linearity to sequences of states on the truncated Fock space $\cFN(\gH)$.
\end{definition}

We emphasize that, by definition, the geometric limit $\Gamma$ must always be a state. That is, it has to satisfy $\Gamma\geq0$ and 
$$\tr_{\cFN}(\Gamma)=G_0+\sum_{n=1}^N\tr_{\gH^n}(G_n)=1.$$
If convergence does not hold in $\gS_1(\gH^N)$, then the final state has to live over spaces with less particles, but its trace is always equal to one. If all the particles are lost, then $\Gamma$ is the vacuum state $\Gamma=1\oplus0\oplus\cdots\oplus0$.
It is proved in~\cite[Lemma 3]{Lewin-11} that every sequence of states $\Gamma_k$ has a geometrically convergent subsequence (the limit being a state).

\subsection*{Main result}
We are now able to state our main result. We consider an abstract  many-body Hamiltonian of the form
\begin{equation}
H_n=\sum_{j=1}^nh_j+\sum_{1\leq j<\ell\leq n}w_{j\ell}
\label{eq:def_H_N}
\end{equation}
acting on the $n$-particle space $\gH^n$ (for us symmetric or antisymmetric). Here $h$ is a given self-adjoint operator acting on the one-particle space $\gH$ and $h_j$ acts on the $j$th factor. On the other hand, $w$ is a self-adjoint operator on the two particle space $\gH^2$ and $w_{j\ell}$ acts on the $j$th and $\ell$th factors. We make rather general assumptions on $w$ in order to give a proper meaning to $H_n$ for all $n$. As will be clear from the rest of the paper, these can be weakened in specific examples or if one is only interested in a particular $n=N$. We assume that
\begin{equation}
 \text{$h$ is bounded from below (without loss of generality $h>1$)}
 \label{eq:hyp_h}
\end{equation}
and that 
\begin{equation}
 \text{$|w|$ is infinitesimally form bounded with respect to $h_1+h_2$}
 \label{eq:hyp_w1}
\end{equation}
which means that 
\begin{equation} \label{eq:w_rel_form_bd}
\epsilon(h_1+h_2)-C_\epsilon \leq w \leq \epsilon(h_1+h_2)+C_\epsilon
\end{equation}
for all $\epsilon>0$. Under these assumptions, $H_n$ is bounded from below and may be realized as a self-adjoint operator for all $n$, by the method of Friedrichs.

\begin{theorem}[Many-body RAGE]\label{thm:mbRAGE}
In addition to~\eqref{eq:hyp_h} and~\eqref{eq:hyp_w1}, assume that for every compact operator $K$ on $\gH$
\begin{equation}
(K_1+K_2) (h_1+h_2)^{-1/2}\,w\, (h_1+h_2)^{-1/2}\\
\label{eq:hyp_w2}
\end{equation}
is compact on $\gH^2$, where $K_1=K\otimes 1_\gH$ and $K_2=1_\gH\otimes K$. Let $\Gamma\geq0$ with $\tr(\Gamma)=1$ be a state on the (symmetric or antisymmetric) $N$-particle space $\gH^N$. Then, the ergodic mean
$$\frac{1}{T}\int_0^T e^{-itH_N}\Gamma\, e^{it H_N}\,dt$$
has geometrically convergent subsequences as $T\to \infty$, and, for every such sequence, the limit 
is a convex combination of projections to eigenspaces of the $n$-body Hamiltonians $H_n$ for $0\leq n\leq N$.
\end{theorem}

If we consider a sequence of times $T_k\to\ii$ for which
$$\frac{1}{T_k}\int_0^{T_k} e^{-itH_N}\Gamma\, e^{it H_N}\,dt\gto M_\infty=G_0\oplus\cdots \oplus G_N,$$
the result states that 
\begin{equation}
G_n=\sum_j\alpha_{n,j}|\phi_{n,j}\rangle\langle\phi_{n,j}| 
 \label{eq:convex_combination}
\end{equation}
where $\{\phi_{n,j}\}_{j\geq1}$ is an orthonormal system of eigenvectors of the $n$-body operator $H_n$. By definition, geometric convergence means that the density matrices converge
$$\frac{1}{T_k}\int_0^{T_k} \Big(e^{-itH_N}\Gamma\, e^{it H_N}\Big)^{(n)}\,dt\underset{\ast}{\wto} M_\infty^{(n)}=\sum_{m=n}^N{m\choose n}\tr_{n+1,...,m}G_m$$
weakly-$\ast$ in the trace-class for all $n=0,...,N$.
Let us emphasize that, because of the partial traces, the $M_\infty^{(n)}$ are in general not supported on the point spectrum of $H_n$. The density matrix $M_\infty^{(n)}$ should not be confused with the restriction $G_n$ of the state $M_\infty$ to the $n$-particle space.

The condition~\eqref{eq:hyp_w2} on $w$ ensures that the interaction between a particle that stays and a particle that escapes vanishes in a weak sense. A condition of this type is clearly necessary to be sure that the remaining $n$ particles are described by the Hamiltonian $H_n$ when the other $N-n$ escape.
The operator $B:=(h_1+h_2)^{-1/2}w(h_1+h_2)^{-1/2}$ is always bounded by Assumption~\eqref{eq:hyp_w1}. However, in the applications, it is usually not compact. The picture is that $w$ only decays in one direction, like the relative coordinate of the two particles, and must be multiplied by a compact operator $K$ in another variable, corresponding to the position of only one of the particles, as in~\eqref{eq:hyp_w2} to make $(K\otimes 1) B$  and $(1\otimes K) B$ compact.

Our theorem is stated for any initial datum $\Gamma$, possibly with an infinite energy. However, our proof does use the conservation of energy for smooth initial data. It is an interesting open problem to derive a similar result when $h$ is not bounded from below. We will make more comments on this below.

We also remark that the exact same theorem holds if the ergodic mean is replaced by 
$$\frac1T\int_\R \chi(t/T)\, e^{-iH_Nt}\Gamma e^{iH_Nt}\,dt$$
where $\chi$ is any nonnegative function such that $\int_\R \chi(t)\,dt=1$.

\subsection*{Reformulation in terms of the CAR and CCR algebras}
Theorem~\ref{thm:mbRAGE} has a natural extension to the Fock space and the associated algebra of canonical (anti-) commutation relations (the latter is to be understood in the sense of Weyl operators, see~\cite{DerGer-13} for a detailed introduction). This extension is based on the observation (see~\cite[Remark 6]{Lewin-11}) that for the rank-one operator
\begin{equation}\label{eq:CRpoly}
 K=\vert f_1\otimes_{s/a}\cdots \otimes_{s/a} f_n\rangle \langle g_1\otimes_{s/a} \cdots \otimes_{s/a} g_n\vert\,,
\end{equation}
geometric convergence of $\Gamma_k\in \gS_1(\cFN)$ just means that
\begin{align}
 \Tr\left( K \Gamma_k^{(n)} \right)&= \left\langle g_1\otimes_{s/a}\cdots \otimes_{s/a} g_n, \Gamma_k^{(n)} f_1\otimes_{s/a}\cdots\otimes_{s/a} f_n \right\rangle\nonumber\\
%
&=\Tr_\cF \left(a^*(f_1)\cdots a^*(f_n) a(g_1)\cdots a(g_n)\Gamma_k\right)\label{eq:geomCR}
\end{align}
converges. Here, $a^*(f), a(f)$ denote, respectively, bosonic or fermionic creation and annihilation operators, depending on whether $\otimes_{s/a}$ is the symmetric or antisymmetric tensor product.
We thus have:
\begin{corollary}
Let $\mathcal{A}\in \lbrace \rm{CCR}(\gH), \rm{CAR}(\gH)\rbrace$ be the $C^*$-subalgebra of the bounded operators on $\cF(\gH)$ satisfying canonical commutation relations, if $\cF(\gH)$ is the symmetric Fock-space, or anti-commutation relations, if $\cF(\gH)$ is the anti-symmetric Fock-space. Denote by $\Pi_N$ the projection of $\cF(\gH)$ to $\gH^N$ and let $\rho\in \mathcal{A}'$ be a normal state on $\mathcal{A}$, that is $\rho(A)=\Tr(\Gamma_\rho A)$ for some $\Gamma_\rho\in \gS_1(\cF)$, satisfying $[\Gamma_\rho,\Pi_N]=0$ for all $N\in \N$.
Define the Hamiltonian on the Fock space $\cF(\gH)$ by
$$\mathbb{H}=0\oplus h\oplus\bigoplus_{N\geq2}H_N.$$
Then the ergodic mean
\begin{equation*}
 \mu_T(A):=\frac1T \int_0^T\rho\left(e^{i \mathbb{H} t} A e^{-i \mathbb{H} t}\right)\,dt\,.
\end{equation*}
has weakly-$\ast$ convergent subsequences as $T\to \infty$ and, for every such sequence, the limit state is normal.
For the density operator $\Gamma_\infty$ of this state, $\Pi_N \Gamma_\infty \Pi_N$ is a convex
combination of projections to eigenspaces of $H_N$ for all $N\geq1$.
\end{corollary}
\begin{proof}
The state $\mu_T$ is normal for all $T>0$, with density matrix 
$$\frac1T \int_0^Te^{-i \mathbb{H} t}\Gamma_\rho\, e^{i \mathbb{H} t}\,dt=\bigoplus_{N\geq0}\frac1T \int_0^Te^{-i H_Nt}\Pi_N\Gamma_\rho\Pi_N\, e^{i H_Nt}\,dt.$$
This sequence of normal states satisfies
$$\tr_{\cF(\gH)}\left\{\bigoplus_{N\geq N_0}\frac1T \int_0^Te^{-i H_Nt}\Pi_N\Gamma_\rho\Pi_N\, e^{i H_Nt}\,dt\right\}=\sum_{N\geq N_0}\tr_{\gH^N}(\Pi_N\Gamma_\rho\Pi_N)$$
which is small, uniformly in $T$, when $N_0$ is large. By an `$\epsilon/2$ argument', it is thus sufficient to prove the corollary for $\Gamma_\rho\in \gS_1(\cFN)$.
For states on the truncated Fock space $\cFN$, the weak-$\ast$ convergence in $\mathcal{A}'$ is equivalent to geometric convergence~\cite[Remark~6]{Lewin-11}. In the fermionic case this is an immediate consequence of the formula~\eqref{eq:geomCR} and the commutation relations. In the bosonic case one uses additionally that on $\gH^N$ the Weyl-operators may be expressed as convergent power series in creation and annihilation operators (see~\cite[chapter 9]{DerGer-13}). Therefore, the existence of weak-$\ast$ convergent subsequences as well as the properties of their limits follow from Theorem~\ref{thm:mbRAGE}, by linearity.
\end{proof}
Note that $\mathcal{A}=\mathrm{CCR}(\gH)$ is not separable, hence bounded sequences do
not necessarily have weakly-$\ast$ convergent subsequences. However, we are
able to prove that our particular sequence $\mu_T$ has convergent
subsequences, using the fact that it is obtained from a fixed normal state by a  particle-number conserving evolution.
\subsection*{Application to Schr\"odinger operators}
A typical example to which our result applies is that of a non-relativistic system of $N$ fermions or bosons, for which the interaction $w$ is a function of the relative position of the particles, that tends to zero at infinity in a weak sense.
\begin{corollary}[Non-relativistic Schr\"odinger operators]\label{cor:Schrodinger}
Let $\gH=L^2(\R^d)$, $h=-\Delta+V(x)+e$ and $w$ be the multiplication operator by an even function $w(x_1-x_2)$. We assume that   $V=f_1+f_2$ and that $w=f_3+f_4$ where $f_i\in L^{p_i}(\R^d)$ for some $\max(1,d/2)<p_i<\ii$, or $f_i\in L^\ii(\R^d)$ and $f_i\to0$ at infinity. Then Theorem~\ref{thm:mbRAGE} holds for the many-body Hamiltonian
\begin{equation*}
 H_N=\sum_{j=1}^N-\Delta_{x_j}+V(x_j)+\sum_{1\leq j<\ell\leq N} w(x_j-x_\ell)\,.
\end{equation*}
\end{corollary}

\begin{proof}
Under the assumptions of the corollary, $(1-\Delta)^{-1/2}V(x)(1-\Delta)^{-1/2}$ and $(1-\Delta)^{-1/2}w(x)(1-\Delta)^{-1/2}$ are compact on $L^2(\R^d)$~\cite[Chap~8]{Davies}, hence the hypothesis~\eqref{eq:hyp_h},~\eqref{eq:hyp_w1} are satisfied and $h>1$ for an appropriate choice of $e$.

The verification of~\eqref{eq:hyp_w2} is more involved, the intuition however is rather simple. It relies on the fact that $(h_1+h_2)^{-1/2}w (h_1+h_2)^{-1/2}$ is compact in the relative coordinate $x_1-x_2$ parametrizing the subspace $V=\lbrace (x_1,x_2)\in \R^{2d}: x_1+x_2=0 \rbrace$, i.e. its action on the space $L^2(V)$ is compact.
This, together with the fact that $K\otimes 1$ is compact in the direction $\R^d\times \lbrace 0\rbrace$ parametrized by $x_1$ and $(\R^d\times \lbrace 0\rbrace) \oplus V=\R^{2d}$, implies that the product is compact. More precisely,
the property we use is the content of the following lemma, which we prove in Appendix~\ref{app:proof_lemma_Georgescu}.
\begin{lemma}\label{lem:semicomp}
 There exist compact operators $K_j$ on $L^2(V)$ and bounded operators $B_j$ on $L^2(V^\perp)$ such that
\begin{equation}\label{eq:W semicomp}
(h_1+h_2)^{-1/2}w (h_1+h_2)^{-1/2}=\sum_{j=1}^\infty K_j \otimes_V B_j
\end{equation}
and the sum converges in the operator-norm.
\end{lemma}
The notation $\otimes_V$ emphasizes that the tensor product of operators is induced by $L^2(\R^{2d})=L^2(V)\otimes L^2(V^\perp)$.
The property~\eqref{eq:W semicomp} is known as $V$-semicompactness and implies hypothesis~\eqref{eq:hyp_w2} by~\cite[Proposition 9.2.2]{AmrBdMGeo-96}. 

The argument is as follows:
It is clearly sufficient to prove that $(K\otimes 1)(K_j\otimes_V 1)$ is compact for every $j$. Then, since Hilbert-Schmidt operators are dense in the compact operators, it suffices to show that $C:=(A\otimes 1)(B\otimes_V 1)$ is compact (and actually Hilbert-Schmidt) for operators $A\in \mathfrak{S}_2(L^2(\R^d))$, $B\in \mathfrak{S}_2(L^2(V))$. 
To prove this, let $a(x,x')$, $b(v,v')$ be the integral kernels of $A$ and $B$, respectively. The operator $C$ acts of $\psi\in L^2(\R^{d}_{x_1}\times \R^d_{x_2})$ as 
\begin{align*}
 (C\psi)(x_1,x_2)
 =\int_{\R^{2d}} a(x_1, x') b(x'-x_2, x'-y')\psi(x', y')\, dx'dy'\,.
\end{align*}
So $C$ is an integral operator, whose kernel is easily seen to be in $L^2(\R^{4d})$. 
\end{proof}

If $w$ and $V$ decay fast enough at infinity, a precise theory of scattering is available~\cite{ReeSim3,Derezinski-93,SigSof-94,HunSig-00b}.
The existence of the wave operators and their completeness implies that there exists vectors $\psi_{n,j}$ such that
\begin{equation}
\lim_{t\to+\ii} \norm{e^{-itH_N}\Psi-\sum_{n=0}^N\sum_{j} e^{-i\lambda_{n,j}t}\big( e^{it\Delta_{\R^{d(N-n)}}}\psi_{n,j}\big)\otimes\phi_{n,j}}_{L^2(\R^{dN})}=0
\label{eq:scattering}
\end{equation}
where $H_n\phi_{n,j}=\lambda_{n,j}\phi_{n,j}$. For the ergodic mean of the time-evolved projection $\vert \Psi \rangle \langle \Psi \vert$ we obtain
\begin{equation*}
\frac1T \int_0^T e^{-itH_N}|\Psi\rangle\langle \Psi|\,e^{itH_N} dt
\gto \bigoplus_{n=0}^N\left(\sum_j \alpha_{n,j} |\phi_{n,j}\rangle\langle \phi_{n,j}|\right)
\end{equation*}
geometrically, where $\alpha_{n,j}=\norm{\psi_{n,j}}^2$  gives the coefficients in~\eqref{eq:convex_combination}. Additionally, any combination of the eigenvectors $\phi_{j,n}$ of the $n$-particle Hamiltonians $H_n$ can occur in the limit, so any $G_0,...,G_N$ of the form~\eqref{eq:convex_combination} may be obtained in the geometric limit, by choosing an appropriate initial condition. We do not know if the same property holds under the more general assumptions of Theorem~\ref{thm:mbRAGE}.

Results similar to Corollary~\ref{cor:Schrodinger} hold for Schr\"odinger operators with magnetic fields, for pseudo-relativistic operators, etc. We do not state them here for shortness. 

The rest of the paper will be devoted to the proof of Theorem~\ref{thm:mbRAGE}.

\section{Proof of Theorem~\ref{thm:mbRAGE}}
Let $\Gamma$ be a state on the $N$-particle space $\mathfrak{H}^N$ and $T_k{\to}\infty$ be a sequence of times such that the ergodic means
\begin{equation*}
M(T_k):=\frac1{T}_k \int_0^{T_k}e^{-i tH_N  } \Gamma e^{i  tH_N}dt
\end{equation*}
with initial condition $\Gamma$ converge geometrically to a limit $M_\infty= \oplus_{n=0}^N G_n$. Such a sequence exists for every $\Gamma$ by~\cite[Lemma 3]{Lewin-11}.

The proof of Theorem~\ref{thm:mbRAGE} proceeds by showing that $M_\infty$ is left invariant by the Hamiltonian $\mathbb{H}:=\oplus_{n=0}^N H_n$ (with $H_0=0$ and $H_1=h$) on the truncated Fock space $\cFN=\C\oplus_{n=1}^N \gH^n$. That is, we have $M_\infty=e^{-i\mathbb{H}t}M_\infty e^{i\mathbb{H}t}$ and thus
$
 e^{-i t H_n} G_n e^{it H_n} \equiv G_n
$,
for every $t\in \R$. Consequently, the eigenspaces of $G_n$ are $H_n$-invariant. 
As $G_n$ is trace-class, the eigenspace corresponding to a non-zero eigenvalue of $G_n$ has finite dimension, and is thus a direct sum of eigenspaces of $H_n$.
\subsection*{Step 1. }
The first step to proving invariance of $M_\infty$ is to note that, for every $s\in\R$, the sequence $e^{-i sH_N  } M(T_k)e^{i  sH_N}$ also converges to $M_\infty$ geometrically. This holds because geometric convergence is controlled by the trace-norm and
 \begin{align}
  e^{-i sH_N } &M(T_k) e^{i sH_N}=
  \frac1{T_k}\int_0^{T_k} e^{-i (t+s)H_N}\Gamma e^{i (t+s)H_N} \,d t\notag\\
  &=\frac1{T_k}\int_s^{T_k+s} e^{-i tH_N}\Gamma e^{i t H_N} \,d t\notag\\
  &=M(T_k)
  + \underbrace{\frac1{T_k} \left(\int_{T_k}^{T_n+s} e^{-i t H_N}\Gamma e^{i  t H_N}d t -\int_0^{s} e^{-i t H_N}\Gamma e^{i  t H_N} d t\right)}_{\leq 2s/T_{k} \text{ in }\mathfrak{S}_1}
 \notag\\
  &\gto M_\infty\,.\label{eq:inv1}
 \end{align}
 Having established this, we would like to prove that also 
\begin{equation}\label{eq:conv time evolv}
 e^{-i sH_N } M(T_k) e^{i sH_N}
 \gto e^{-i s\mathbb{H} } M_\infty e^{i s\mathbb{H}}\,.
\end{equation}
As the left hand side depends only on $H_N$, which equals the restriction of $\mathbb{H}$ to $\gH^N$, and the right may depend on $\mathbb{H}$ on all the sectors with $n\leq N$ particles, this will certainly not be true for arbitrary Hamiltonians on Fock space. In fact, the proof of~\eqref{eq:conv time evolv} will depend crucially both on the properties of $\mathbb{H}$ and the sequence $M(T_k)$.
\subsection*{Step 2. }
Using the particular form of the sequence $M(T_k)$ we can reduce the proof of~\eqref{eq:conv time evolv} to sequences of bounded energy. That is, let $e_N$ be such that $H_N+e_N>1$ and 
assume the initial condition $M(0)=\Gamma$ satisfies
\begin{equation*}
 \Tr\left((H_N+e_N)^{1/2}\Gamma (H_N+e_N)^{1/2}\right)
 <E
\end{equation*}
for some constant $E>0$.
Then, since $e^{-itH_N}$ commutes with $H_N$ and preserves the trace-norm, we have 
\begin{equation}\label{eq:unif en bound}
 \Tr\left((H_N+e_N)^{1/2}M(T_k) (H_N+e_N)^{1/2}\right)<E\,,
\end{equation}
for every $k\in \N$.
Since any initial condition $M(0)=\Gamma$ can be approximated to arbitrary precision by states of finite energy, it is sufficient to prove~\eqref{eq:conv time evolv} for sequences satisfying~\eqref{eq:unif en bound}, by virtue of the following lemma.
\begin{lemma}
 There exists a constant $c$, depending only on $N$, such that for every sequence $\lbrace \gamma_k\rbrace_{k=1}^\infty$ in $\gS_1(\gH^N)$ that converges geometrically to $\gamma \in \gS_1(\cFN)$ and satisfies $\norm{\gamma_k}_{\gS_1(\gH^N)}< \eps$ we have
 $\norm{\gamma}_{\gS_1(\cFN)}\leq c\eps$.
\end{lemma}
\begin{proof}
Since the trace-class is the dual of the compact operators,
we have
\begin{align*}
\norm{\gamma^{(n)}}_{\gS_1(\gH^n)} 
&= \sup_{K \text{compact, } \norm{K}=1} \lim_{k\to \infty} \left\vert \Tr_{\gH^n}\left(K \gamma_k^{(n)}\right)\right\vert  \\
&= {N \choose n}\sup_{\norm{K}=1} \lim_{k\to \infty} \left\vert \Tr_{\gH^N} \left( K\otimes 1_{N-n} \gamma_k \right)\right\vert 
\leq {N \choose n} \eps\,.
\end{align*}
Now $\gamma=\oplus_{n=0}^N G_n$ is completely determined by its reduced density matrices (cf.~\cite[Lemma 1]{Lewin-11}), explicitly
\begin{equation*}
G_n=\gamma^{(n)} + \sum_{j=1}^{N-n}(-1)^j {n+j \choose n} \Tr_{n+1,\dots,n+j} \gamma^{(n+j)}\,,
\end{equation*}
which proves the claim.
\end{proof}
It will often be useful to state the energy bound using a Hamiltonian without interaction. Let
\begin{equation*}
H^0_n:=\sum_{j=1}^n h_j>1
\end{equation*}
and
\begin{equation*}
\gS_{1,H^0}(\gH^n):=\left\lbrace \gamma\in \gS_1(\gH^n): \big\Vert \sqrt{H^0_n} \gamma\sqrt{ H^0_n}\big\Vert_{\gS_1}<\infty \right\rbrace\,.
\end{equation*}
This space has a natural norm given by $\norm{\gamma}:=\norm{\gamma}_{\gS_1} + \Vert\sqrt{H^0_n} \gamma \sqrt{H^0_n}\Vert_{\gS_1}$.
Using $H_n$ instead of $H^0_n$ gives an equivalent norm, because of the inequalities~\eqref{eq:w_rel_form_bd}. We also define
\begin{equation*}
\gS_{1,\mathbb{H}^0}(\cFN):=\left\lbrace \gamma\in \gS_1(\cFN): \big\Vert\sqrt{\mathbb{H}^0} \gamma \sqrt{\mathbb{H}^0}\big\Vert_{\gS_1}<\infty \right\rbrace\,,
\end{equation*}
with $\mathbb{H}^0=\oplus_{n=0}^N H^0_n$, $H^0_0=0$.
It will be important that geometric convergence preserves such energy estimates. The following lemma proves that the unit ball of $\gS_{1,\mathbb{H}^0}(\cFN)$ is closed under geometric convergence. This is the only step in the proof of Theorem~\ref{thm:mbRAGE} for which positivity of $H^0_n$ is essential.
\begin{lemma}\label{lem:energynorm}
Let $\lbrace \gamma_k\rbrace_{k=1}^\infty$ in $\gS_{1,\mathbb{H}^0}(\cFN)$ be a bounded sequence that converges geometrically to $\gamma$. Then for every $n\leq N$
\begin{equation*}
\sqrt{H^0_n} \gamma_k^{(n)} \sqrt{H^0_n} \wto_* \sqrt{H^0_n}\gamma^{(n)} \sqrt{H^0_n}
\end{equation*}
in $\gS_1(\gH^n)$.
\end{lemma}
\begin{proof}
By linearity, it is sufficient to prove the claim for $\gamma_k\in \gS_{1,H^0}(\gH^N)$.
For any $n\leq N$ we have
\begin{equation*}
0\leq (H^0_n \otimes 1_{\gH^{N-n}} ) (H^0_N)^{-1} =1 - (1_{\gH^{n}}\otimes H^0_{N-n}  ) (H^0_N)^{-1}\leq 1\,,
\end{equation*}
because $H_n^0\geq1$.
Thus, for an arbitrary $\eta \in \gS_{1,H^0}(\gH^N)$,
\begin{align}
 \norm{\sqrt{H^0_n}\eta^{(n)} \sqrt{H^0_n}}_{\gS_1} 
 &= \sup_{B \text{ bounded, }\norm{B}=1} \left\vert \Tr\left(B \sqrt{H^0_n}\eta^{(n)} \sqrt{H^0_n}\right) \right \vert\notag\\
 %
 %
 &\leq {N \choose n} \norm{(\sqrt{H^0_n}\otimes 1) \eta (\sqrt{H^0_n}\otimes1)}_{\gS_1}\notag\\
 &\leq {N \choose n} \norm{\sqrt{H^0_N} \eta \sqrt{H^0_N}}_{\gS_1}\,.\label{eq:energynorm}
\end{align}
So $\sqrt{H^0_n} \gamma_k^{(n)}\sqrt{H^0_n}$ is uniformly bounded in $\gS_1(\gH^n)$ and there exists a subsequence such that
\begin{equation*}
 \sqrt{H^0_n} \gamma_k^{(n)} \sqrt{H^0_n} \wto_* \tilde \gamma^{(n)}\,.
\end{equation*}
Now for every compact operator $K$ on $\gH^n$ we have
\begin{align*}
\Tr\big( (H^0_n)^{-1/2} K (H^0_n)^{-1/2} \tilde \gamma^{(n)} \big)
&=\lim_{k\to \infty} \Tr\big( (H^0_n)^{-1/2} K \gamma_k^{(n)} (H^0_n)^{1/2} \big)\\
&=\Tr\big(K \gamma^{(n)}\big)\,,
\end{align*}
whence $\tilde \gamma^{(n)} = \sqrt{H^0_n} \gamma^{(n)}\sqrt{H^0_n}$. 
\end{proof}
\subsection*{Step 3. }
The reduction of the problem to sequences of bounded energy in step 2 will now allow us to study the sequence $e^{-i s H_N} M(T_k) e^{i s H_N}$ via the differential equation it satisfies. If $\gamma(t)$ is a solution to the von Neumann equation~\eqref{eq:Heisenberg}, its reduced density matrices (formally) satisfy the finite BBGKY hierarchy
\begin{equation*}
 i \frac{d}{dt}\gamma^{(n)}(t)=[H_n,\gamma^{(n)}(t)] +  (n+1)\sum_{j=1}^n\Tr_{n+1}\left([w_{j,n+1}, \gamma^{(n+1)}(t)]\right)\,.
\end{equation*}
Note, however, that the equation has no clear meaning (not even in a weak sense) if $w$ is not bounded, due to the partial trace. If $w$ is a bounded operator $D((H_2^0)^\alpha)\to \gH^2$ for some $\alpha>0$, this problem can be handled by considering only initial conditions satisfying $\Tr((H_N^0)^\alpha\gamma (H_N^0)^\alpha)<E$, but to deal with potentials that are only $H_2^0$-form-bounded we will have to define a modified equation, that is equivalent to the original one for bounded $w$.
\begin{proposition}[Well-posedness of the truncated BBGKY hierarchy]\label{thm:BBGKY}\leavevmode\\
\begin{enumerate}
 \item For every $\gamma\in \gS_{1,\mathbb{H}^0}(\cFN)$ 
 the family of reduced density matrices $\gamma(t)^{(n)}=\left(e^{-i t\mathbb{H}}\gamma e^{i t\mathbb{H}}\right)^{(n)}$ is the unique solution $\lbrace \gamma^{(n)}(t): n=0,\dots, N\rbrace$ to the 
 the system of equations
\begin{align}
 & \gamma^{(n)}(t)
\begin{aligned}[t]\notag
 =&e^{-it H_{n}}\gamma^{(n)}e^{it H_{n}}\\
&-i(n+1) \sum_{j=1}^n \int_0^t e^{-i(t-s) H_{n}}\Tr_{n+1}\left(\mathcal{L}_{jn}\left(w,\gamma^{(n+1)}(s)\right)\right)e^{i(t-s) H_{n}}ds,
\end{aligned}\\
 &\mathcal{L}_{jn}\left(w,\gamma^{(n+1)}\right)=
\left[h_{n+1}^{-1/2}w_{j,n+1}h_{n+1}^{-1/2},h_{n+1}^{1/2}\gamma^{(n+1)}(s)h_{n+1}^{1/2}\right]\,,\label{eq:w-duhamel}
\end{align}
such that
\begin{equation*}
 \gamma^{(n)}(t)\in 
 L^\infty\big( \R, \gS_{1,H^0_n}(\gH^n)\big)\,.
\end{equation*}
\item For every bounded sequence $\gamma_k \in \gS_{1,\mathbb{H}^0}(\cFN)$ that converges geometrically to $\gamma_\infty$, the corresponding solutions $\gamma_k(t)$ of the von Neumann equation~\eqref{eq:Heisenberg} with initial condition $\gamma_k$ converge geometrically to the solution with initial condition $\gamma_\infty$:
 \begin{equation*}
 \forall t\in \R:\quad e^{-i t\mathbb{H}}\gamma_k e^{i t\mathbb{H}}
 \gto  e^{-i t\mathbb{H}}\gamma_\infty e^{i t\mathbb{H}}\,.
\end{equation*}
\end{enumerate}
\end{proposition}
\begin{proof}
 Note that $\gamma^{(N+1)}=0$ and this is a triangular system, which can be solved starting with $\gamma^{(N)}(t)= e^{-it H_{N}}\gamma^{(N)}e^{it H_{N}}$. 
 This immediately gives us uniqueness, for assume we have two solutions  with $\gamma^{(n)}(0)=\gamma^{(n)}=\theta^{(n)}(0)$, then $\gamma^{(N)}(t)=\theta^{(N)}(t)$. If $\gamma^{(N)},\theta^{(N)}\in L^\infty\big( \R, \gS_{1,H^0_n}\big)$ this then implies $\gamma^{(N-1)}(t)=\theta^{(N-1)}(t)$.
 
 We now check that $\gamma(t)^{(n)}=\left(e^{-i t\mathbb{H}}\gamma e^{i t\mathbb{H}}\right)^{(n)}$ is indeed a solution. By linearity, we may restrict to initial conditions in $\gamma\in\gS_{1,H^0}(\gH^N)$. Since both $(H_N+e_N)^{1/2} (H_N^0)^{-1/2}$ and $(H_N^0)^{1/2} (H_N+e_N)^{-1/2}$ are bounded, due to hypothesis~\eqref{eq:hyp_w1}, we have $\gamma(t)=\gamma(t)^{(N)}\in L^\infty\big(\R, \gS_{1,H^0}(\gH^N)\big)$.
 Equation~\eqref{eq:energynorm} then implies that $\gamma(t)^{(n)}\in L^\infty\big(\R, \gS_{1,H^0}(\gH^n)\big)$. 
 
Now assume for a moment that $w$ is bounded and let
\begin{equation*}
W_n:=\sum_{1\leq j \leq n< \ell\leq N} w_{j\ell}
\end{equation*}
be the interaction of the first $n$ particles with the remaining $N-n$.
Then we have
\begin{equation*}
 H_N=H_n\otimes 1_{\gH^{N-n}} + 1_{\gH^{n}}\otimes H_{N-n} + W_n\,,
\end{equation*}
and we can write $\gamma(t)^{(N)}$ using Duhamel's formula,
\begin{align*}
 \gamma(t)^{(N)}=&e^{-it H_{N-n}}e^{-it H_n}\gamma e^{it H_n}e^{it H_{N-n}}\\
 &-i \int_0^t e^{-i(t-s)  H_{N-n}}e^{-i(t-s) H_n}\left[ W_n, \gamma(s)^{(N)}\right]e^{i(t-s) H_{N-n}}e^{i(t-s) H_n}ds
\end{align*}
(with the $1\otimes$ omitted for shortness). This implies that
\begin{align*}
 &\gamma(t)^{(n)}-e^{-it H_n}\gamma^{(n)}e^{it H_n}\\
 &\quad=
-i{N \choose n}\int_0^te^{-i(t-s) H_n}\Tr_{n+1, \dots, N}\left(\left[ W_n, \gamma(s)^{(N)}\right]\right)e^{i(t-s) H_n}ds\\
&\quad=
-i(n+1)\int_0^t\sum_{j=1}^ne^{-i(t-s) H_n}\Tr_{n+1}\left(\left[w_{j,n+1},\gamma(s)^{(n+1)}\right]\right)e^{i(t-s) H_n}ds\,.
\end{align*}
As $h_{n+1}^{1/2}\gamma(s)^{(n+1)}h_{n+1}^{1/2}$ is trace-class and $w$ is bounded, we have
\begin{align*}
\Tr_{n+1}\left(\left[w_{j,n+1},\gamma^{(n+1)}\right]\right)=
\Tr_{n+1}\left(\left[h_{n+1}^{-1/2}w_{j,n+1}h_{n+1}^{-1/2},h_{n+1}^{1/2}\gamma^{(n+1)}h_{n+1}^{1/2}\right]\right)\,,
\end{align*}
which gives~\eqref{eq:w-duhamel}. In order to account for unbounded interactions, let $w^\delta:=(\delta H_2^0+1)^{-1/2}w (\delta H_2^0+1)^{-1/2}$ for $\delta>0$, which is bounded by hypothesis~\eqref{eq:hyp_w1}. We clearly have
\begin{equation*}
 \lim_{\delta\to 0} (H_2^0)^{-1/2}w^\delta (H_2^0)^{-1/2}= (H_2^0)^{-1/2}w (H_2^0)^{-1/2}
\end{equation*}
in the strong topology of operators on $\gH^2$. Hence, the constant $C_\eps$ of equation~\eqref{eq:w_rel_form_bd} may be chosen in such a way that the inequalities hold for all $0<\delta\leq 1$. Then, for an appropriate choice of $e_N\geq0$,
\begin{equation*}
 \Big( H_N^0 + \sum_{1\leq j<\ell\leq N} w^\delta_{j\ell}+e_N \Big)^{-1/2}:\gH^N\to D\big((H_N^0)^{1/2}\big)
\end{equation*}
is bounded, uniformly in $\delta$. This implies that, as operators on $D\big((H_N^0)^{1/2}\big)$,
\begin{equation*}
\lim_{\delta \to 0} \Big( H_N^0 + \sum_{1\leq j<\ell \leq N} w^\delta_{j\ell} + e_N\Big)^{-1}=\left( H_N +e_N\right)^{-1}
\end{equation*}
strongly. This in turn implies that the unitary groups $U_\delta$ generated by these operators also converge in the strong operator topology on $D\big((H_N^0)^{1/2}\big)$. So, for $\gamma\in \gS_{1,H^0}$, we have
\begin{equation*}
 \lim_{\delta \to 0} (H_N^0)^{1/2}U_\delta(t) \gamma = (H_N^0)^{1/2} e^{-it H_N}\gamma 
\end{equation*}
in trace-norm. We thus have, using~\eqref{eq:energynorm}, that
\begin{equation*}
 \lim_{\delta\to 0} \left(U_\delta(t) \gamma U_\delta(t)^*\right)^{(n)}=\gamma^{(n)}(t) 
\end{equation*}
in $\gS_{1,H^0}(\gH^n)$. As $w^\delta$ is bounded, the left hand side solves~\eqref{eq:w-duhamel} with this interaction.
To take the limit on the right hand side of~\eqref{eq:w-duhamel} observe that
\begin{equation*}
 \lim_{\delta\to 0} \Tr_{n+1}\mathcal{L}_{jn}\left((w^\delta, \gamma^{(n+1)})\right)=\Tr_{n+1}\mathcal{L}_{jn}\left((w, \gamma^{(n+1)})\right)
\end{equation*}
in the space of operators 
for which $(H_n^0)^{-1/2}A(H_n^0)^{-1/2}$ is trace-class.
As $\gS_{1,H^0}$ is obviously contained in this space, this shows that $\gamma(t)^{(n)}$ is indeed a solution.

We now prove the continuous dependence on the initial condition, item~\emph{(2)}. Since we have already proved uniqueness, it is sufficient to show that the geometric limit of $\gamma_k(t)=e^{-it \mathbb{H}}\gamma_k e^{it \mathbb{H}}$ is a solution of~\eqref{eq:w-duhamel}.
On the space $\gH^N$ this is trivial, since for any compact operator $K$, $K(t):=e^{itH_N} K e^{-itH_N}$ is also compact and thus
\begin{equation*}
 \lim_{k\to \infty} \Tr\left(K e^{-itH_N}\gamma_k e^{itH_N}\right)= \Tr\left(K e^{-itH_N}\gamma_\infty e^{itH_N}\right)\,,
\end{equation*}
for every $t\in \R$.
Now let $m<N$ and assume that $\gamma_k(t)^{(n)}\wto_* \gamma_\infty(t)^{(n)}$ in $\gS_1(\gH^n)$ for every $t\in \R$ and $n=m+1,\dots  N$.
Then, we have for every compact operator $K$ on $\gH^m$
\begin{align*}
&\frac{i}{m+1}\left( \tr_{\gH^m}\left(K \gamma_k(t)^{(m)}\right)-\Tr_{\gH^m}\left( e^{-itH_m}K\gamma_k(0)^{(m)} e^{itH_m}\right)\right)\\
&=\int_0^t \sum_{j=1}^m\Tr_{\gH^{m+1}}\left((K(t-s)\otimes_{s/a} 1)\mathcal{L}_{jm}(w,\gamma^{(m+1)})
\right)\,ds \,.
\end{align*}
It is clearly sufficient to prove convergence for $K$ in a dense set of compact operators, so we may assume that $K(t-s)=\tilde K(t-s) (H_m)^{-1/2}$ with compact $\tilde K$. For such an operator $K$, the integrand is uniformly bounded, and we will show that it converges pointwise.
Using Lemma~\ref{lem:energynorm}, the induction hypothesis gives us
\begin{equation*}
(H_{m+1}^0)^{1/2}\gamma_k^{(m+1)}(s)(H_{m+1}^0)^{1/2}\wto_*(H_{m+1}^0)^{1/2}\gamma^{(m+1)}_\infty(s)(H_{m+1}^0)^{1/2}\,.
\end{equation*}
It is thus enough to prove that
\begin{equation*}
 (\tilde K_j(t-s)\otimes_{s/a} 1)(h_{m+1}h_j)^{-1/2} w_{j,m+1}(H_{m+1}^0)^{-1/2}\,,
\end{equation*}
 where $\tilde K_j(t-s)= \tilde K(t-s) (H_m^0)^{-1/2} h_j^{1/2}$, is compact. 
Hypothesis~\eqref{eq:hyp_w2} guarantees that 
 \begin{equation*}
 (L\otimes_{s/a} 1)(h_{m+1}+h_j)^{-1/2} w_{j,m+1} (H^0_{m+1})^{-1/2}
 \end{equation*}
is compact for any compact operator $L$. Hence, it is sufficient to show that 
\begin{equation*}
 (\tilde K_j(t-s)\otimes 1)(h_{m+1}h_j)^{-1/2}(h_j+h_{m+1})^{1/2}= \sum_{l=1}^\infty L_l\otimes B_l\,,
\end{equation*}
as a norm-convergent sum, with bounded operators $B_l$ on $\gH$ and compact operators $L_l$ on $\gH^m$. Since $h_j$ and $h_{m+1}$ commute, we have
\begin{equation*}
 \left( (h_{m+1}+h_j)/h_jh_{m+1}\right)^{-1/2}=\left(1+h_j^{-1}\right)^{-1/2}\left(1-\frac{1-h_m^{-1}}{1+h_j^{-1}}\right)^{-1/2}\,,
\end{equation*}
and because $0<(1-h_m^{-1})/(1+h_j^{-1})<1$ we can write this using the convergent power series of $(1-x)^{-1/2}$ on $\vert x\vert <1$, which gives the desired form.

We have thus shown that $\gamma_k(t)^{(m)}$ converges weakly-$\ast$ to the right hand side of equation~\eqref{eq:w-duhamel} with $\gamma^{(m+j)}(t)=\big(e^{-it \mathbb{H}}\gamma_\infty e^{it \mathbb{H}}\big)^{(m+j)}$, which proves the claim by the uniqueness of solutions to~\eqref{eq:w-duhamel}.
\end{proof}
This completes the proof of Theorem~\ref{thm:mbRAGE}. To summarize: 
We established that the Theorem is implied by invariance of the limit state $M_\infty$ under $\mathbb{H}$.
As a first step~\eqref{eq:inv1}, we then showed that this holds if the limit of the time-evolved sequence is the evolved limit state (equation~\eqref{eq:conv time evolv}).
In the second step, we reduced the problem to considering initial conditions with bounded energy, and thus to proving~\eqref{eq:conv time evolv} for sequences of uniformly bounded energy, i.e.~satisfying~\eqref{eq:unif en bound}.
The final step consisted in studying the BBGKY-hierarchy satisfied by the reduced density matrices under the condition of finite energy, in its weak form~\eqref{eq:w-duhamel}. We proved uniqueness of solutions to this equation in Proposition~\ref{thm:BBGKY}. Additionally, we proved that that geometric convergence for the initial conditions, together with a uniform energy bound, implies geometric convergence of the corresponding solutions at any finite time. 
For the sequence $M(T_k)$ of ergodic means this gives
\begin{equation*}
0_{\cF^{\leq N-1}}\oplus e^{-i H_N t} M(T_k)e^{i H_N t} = e^{- i\mathbb{H} t} (0\oplus M(T_k)) e^{i\mathbb{H} t}
\gto e^{-i t\mathbb{H}}M_\infty e^{i t\mathbb{H}}\,,
\end{equation*}
that is~\eqref{eq:conv time evolv}.
As we have shown, this implies invariance of the limit state
\begin{equation*}
 e^{-i t\mathbb{H}}M_\infty e^{i t\mathbb{H}}= \glim\limits_{k\to \infty} e^{-i H_N t} M(T_k)e^{i H_N t} \stackrel{\eqref{eq:inv1}}{=} M_\infty\,,
\end{equation*}
and thus Theorem~\ref{thm:mbRAGE}.
\begin{remark}
 Instead of studying the evolution equation for the density matrices as in step 3, one could also study the dual evolution on the space of operators that we test against. The condition that needs to be verified is that this space is invariant under conjugation by $e^{-it\mathbb{H}}$. 

Geometric convergence of $\gamma_k$ is defined via convergence of $\Tr((K\otimes_{s/a} 1) \gamma_k)$, so the natural space of operators is $K\otimes_{s/a} 1$, where $K$ is compact on $\gH^n$ for some $n\leq N$. However, this space will in general not be invariant under $e^{-it\mathbb{H}}$ (e.g.~ if $w$ is not bounded, the time-derivative is never tangent to this space).
This amounts to the fact that the map $\gamma\mapsto e^{-it\mathbb{H}}\gamma e^{it\mathbb{H}}$ is not continuous with respect to the geometric topology on $\gS_1(\cFN)$.
However, Proposition~\ref{thm:BBGKY} proves that this map is continuous for the induced topology on bounded sets of $\gS_{1,\mathbb{H}^0}(\cFN)$. These sets are closed by Lemma~\ref{lem:energynorm}. 
The dual of this space is spanned by the operators $A\otimes_{s/a} 1$ such that $(H_n^0)^{-1/2}A (H_n^0)^{-1/2}$ is compact, and it can be shown to be invariant by a reasoning similar to that applied in part two of Proposition~\ref{thm:BBGKY}. 
 
This is sufficient to obtain Theorem~\ref{thm:mbRAGE} for arbitrary initial conditions, since the sequence $M(T_k)$ can be uniformly approximated in that space.
\end{remark}

\appendix
\section{Proof of Lemma~\ref{lem:semicomp}}\label{app:proof_lemma_Georgescu}
Recall the notation $V:=\lbrace (x_1, x_2)\in \R^{d}\times\R^{d}: x_1+x_2=0\rbrace$ and parametrize this subspace by the coordinate $v=x_1-x_2$. 
We will begin by showing that
\begin{equation}\label{eq:WD semicomp}
 (1-\Delta_{x_1}-\Delta_{x_2})^{-1/2}w(v)(1-\Delta_{v})^{-1/2}=\sum_{j=1}^\infty K_j \otimes_V B_j\,, 
\end{equation}
with compact $K_j$ and bounded $B_j$.

Let $\xi$ denote the conjugate Fourier-variable to $v$ and $\zeta$ that to $v^\perp:=x_2+x_1$. Let $\eta(\xi,\zeta)=(1+2\xi^2+2\zeta^2)^{-1/2}$, so $(1-\Delta)^{-1/2}$ is just the Fourier-multiplier by $\eta$.
Since $\eta$ tends to zero at infinity, there exist functions $f_j(\xi), g_j(\zeta)$ with compact support, such that $\sum_{j=1}^m f_j(\xi)g_j(\zeta)\stackrel{k\to \infty}{\to}\eta(\xi, \zeta)$ in $L^\infty$. We can additionally arrange to have $\sum_{j=1}^k f_j(\xi)g_j(\zeta)\leq\eta(\xi, \zeta)$ for every $k$.
Now define $B_j:=\mathcal{F}^{-1}_\zeta g_j(\zeta)\mathcal{F}_{v^\perp}$ and $K_j:=\mathcal{F}^{-1}_\xi f_j(\xi)\mathcal{F}_v w(v)(1-\Delta_{v})^{-1/2} $, where $\mathcal{F}_x$ is the Fourier transform in the variable $x$.
These operators clearly have the desired properties, and we now need to check convergence of the sum~\eqref{eq:WD semicomp}, which is not immediately obvious if $w$ is not bounded. 
First, note that $K_0:= \mathcal{F}^{-1}_{\xi}(1+\xi^2)^{-1/2}\mathcal{F}_v w(v)(1-\Delta_{v})^{-1/2}$ is a compact operator on $L^2(V)$ and $f_n \sqrt{1+\xi^2}$ is bounded. 
In order to exploit the fact that $K_0\otimes_V 1$ is compact in the first factor, we identify the operators on $L^2(V)\otimes L^2(V^\perp)$ with operators from $L^2(V)$ to $\mathcal{L}\big(L^2(V^\perp), L^2(\R^{2d})\big)$ via $A(\psi)\varphi:=A \psi\otimes \varphi$.
After this identification, $K_0\otimes_V 1$ defines a compact operator, since for any weakly convergent sequence $\psi_k \wto \psi$ in $L^2(V)$ and $\phi\in L^2(V^\perp)$
\begin{equation*}
 \norm{(K\otimes_V 1)(\psi_k)\varphi - (K\otimes 1)(\psi)\varphi}\leq \norm{K\psi_k -K\psi}_{L^2(V)} \norm{\varphi}_{L^2(V^\perp)}\,.
\end{equation*}
Furthermore, the operators
\begin{equation*}
A_m:=\sum_{j=1}^m \mathcal{F}^{-1} f_j(\xi)\sqrt{1+\xi^2} g_j(\zeta)\mathcal{F}
\end{equation*}
converge to $\mathcal{F}^{-1}\eta(\xi,\zeta)\sqrt{1+\xi^2}\mathcal{F}$ in the strong operator topology of\\ $\mathcal{L}\big(L^2(V),\mathcal{L}(L^2(V^\perp), L^2(\R^{2d}))\big)$, since for every $\psi\in L^2(V)$
\begin{align*}
&\lim_{m\to \infty}\norm{A_m(\psi) - \mathcal{F}^{-1}_{(\xi, \zeta)}\eta(\xi,\zeta)\sqrt{1+\xi^2}\mathcal{F}_{(v, v^\perp)}(\psi)}_{\mathcal{L}(L^2(V^\perp), L^2(\R^{2d}))}\\
&= \lim_{m\to \infty} \sup_{\zeta\in \R^d} \int_{\R^d} \bigg \vert \Big(\sum_{j=1}^m f_j(\xi)g_j(\zeta) -\eta(\xi,\zeta)\Big)\sqrt{1+\xi^2} \hat\psi(\xi)\bigg\vert^2 \,d\xi =0\,,
\end{align*}
as the integrand is bounded by $\vert 2\eta\sqrt{1+\xi^2} \hat\psi \vert^2\in L^1$ and converges to zero, pointwise in $\xi$ and uniformly in $\zeta$.
Thus, since $\sum^m K_j\otimes_V B_j= A_m (K_0\otimes_V 1)$, the sum in~\eqref{eq:WD semicomp} converges in norm.
To deduce the statement of lemma~\ref{lem:semicomp} from~\eqref{eq:WD semicomp} we use use the formula
\begin{align}
 &(h_1+h_2)^{-1/2}-(2e-\Delta_{x_1}-\Delta_{x_2})^{-1/2}\label{eq:comp int}\\
 &=\frac{1}{2\pi i} \int_{\sigma} z^{-1/2}(h_1+h_2-z)^{-1}\left(V(x_1)+ V(x_2)\right)(2e-\Delta_{x_1}-\Delta_{x_2}-z)^{-1}\,dz\notag, 
\end{align}
where $\sigma$ is the boundary of a sector in the right half plane $\mathrm{Re}(z)>0$ with $\sigma\cap \lbrace \mathrm{Im}(z)=0\rbrace =\lbrace e\rbrace$ and the integral converges in the operator norm. The argument we used to prove~\eqref{eq:WD semicomp} also implies that $(1-\Delta_{x_1}-z)^{-1/2}V(x_1)(2e-\Delta_{x_1}-\Delta_{x_2}-z)^{-1/2}$ is $x_1$-semicompact, and thus (cf. Lemma~\ref{cor:Schrodinger}) 
\begin{equation*}
 (1-\Delta_{x_1}-z)^{-1/2}V(x_1)(2e-\Delta_{x_1}-\Delta_{x_2}-z)^{-1}w(v)(1-\Delta_{v})^{-1/2}
\end{equation*}
is a compact operator for every $z\in \sigma$. Now for $f\in L^p(\R^d)$ with $p>\mathrm{max}(1,d/2)$ the Kato-Seiler-Simon inequality gives
\begin{equation*}
 \norm{(1-\Delta_{x}-z)^{-1/2}f(x)(1-\Delta_{x}-z)^{-1/2}}_{\gS_p}\hspace{-8pt}
 \leq C(p,d) (1+\vert z \vert^2)^{ (d/4p-1/2)}\norm{f}_{L^p}.
\end{equation*}
Using the decomposition $V=f_1+f_2$ we see that the integral in~\eqref{eq:comp int} still converges in norm if multiplied by $w(v)(1-\Delta_{v})^{-1/2}$ from the right. Consequently, it defines a compact operator, so the difference between equations~\eqref{eq:W semicomp} and~\eqref{eq:WD semicomp} is compact and the proof is complete.

\bigskip

\noindent\textbf{Acknowledgment.} We thank Laurent Bruneau, Vladimir Georgescu and Phan Th\`anh Nam for useful discussions. We acknowledge financial support from the European Research Council under the European Community's Seventh Framework Programme (FP7/2007-2013 Grant Agreement MNIQS 258023).


\end{document}